\documentclass[11pt]{article}
\pdfoutput=1
\usepackage[%
compact=,%
paper=a4paper,%
largepaper=true,%
marginfrac=12,%
runningtitle=false,%
halfparskip,%
font=Palatino,%
biblatex=false,%
natbib=true,%
bibliographystyle=plainurl,%
theorems=numbersfirst,%
theoremswithin=document,%
]{ifiseries}
\usepackage[graphtheory,gametheory,complexitytheory]{newKit}
\RenewDocumentCommand\pay{}{\mu}
\NewDocumentCommand\change{mmm}{#1[#2 \leftarrow #3]}
\NewDocumentCommand\apply{mm}{#1\parens[big]{#2}}
\NewDocumentCommand\applysmall{mm}{#1\parens{#2}}
\NewDocumentCommand\DIS{}{\cD}
\NewDocumentCommand\nn{}{\nu}
\NewDocumentCommand\localparam{}{\lam}

\begin{document}
\title{Price of Anarchy for Graph Coloring Games\\
with Concave Payoff}
\author{Lasse Kliemann \and Elmira Shirazi Sheykhdarabadi \and Anand Srivastav}
\date{\smaller%
Department of Computer Science\\Kiel University\\Christian-Albrechts-Platz 4\\24118 Kiel, Germany\\
\texttt{\{lki,esh,asr\}@informatik.uni-kiel.de}}
\maketitle
\enlargethispage{1\baselineskip}

\begin{abstract}
  We study the price of anarchy in a class of graph coloring games (a subclass of polymatrix common-payoff games).
  In those games, players are vertices of an undirected, simple graph,
  and the strategy space of each player is the set of colors
  from $1$ to $k$.
  A tight bound on the price of anarchy of $\frac{k}{k-1}$
  is known (Hoefer 2007, Kun \etal 2013), for the case that
  each player's payoff is the number of her neighbors with different color than herself.
  The study of more complex payoff functions was left as an open problem.
  \par\smallskip
  We compute payoff for a player
  by determining the \emphasis{distance} of her color
  to the color of each of her neighbors,
  applying a non-negative, real-valued, \emphasis{concave} function~$f$ to each of those distances,
  and then summing up the resulting values.
  This includes the payoff functions suggested by Kun \etal (2013) for future work as special cases.
  \par\smallskip
  Denote $f^*$ the maximum value that $f$ attains on the possible distances $0,\hdots,k-1$.
  We prove an upper bound of $2$ on the price of anarchy for concave functions $f$
  that are non-decreasing or which assume $f^*$ at a distance on or below $\floor{\frac{k}{2}}$.
  Matching lower bounds are given for the monotone case and for the case that $f^*$ is assumed in $\frac{k}{2}$ for even $k$.
  For general concave functions, we prove an upper bound of $3$.
  We use a simple but powerful technique:
  we obtain an upper bound of $\localparam \geq 1$ on the price of anarchy
  if we manage to give a \term{splitting} $\localparam_1 + \hdots + \localparam_k = \localparam$ such that
  $\sum_{s=1}^k \localparam_s \cdot f(\abs{s-p}) \geq f^*$ for all $p \in \setft{1}{k}$.
  The discovery of working splittings can be supported by computer experiments.
  We show how, once we have an idea what kind of splittings work,
  this technique helps in giving simple proofs,
  which mainly work by case distinctions, algebraic manipulations, and real calculus.
  \par\smallskip
  Graph coloring, especially in distributed and game-theoretic settings,
  is often used to model spectrum sharing scenarios, such as the selection of WLAN frequencies.
  For such an application, our framework would allow to express a dependence of the degree of radio interference
  on the distance between two frequencies.
\end{abstract}

\section{Model, Notation, Basic Notions}
Let $G=(V,E)$ be an undirected, simple graph without any isolated vertices
($n\df\card{V}$ and $m\df\card{E}$) and $k \in \NN_{\geq 2}$.
For $v \in V$ denote $N(v) \df \set{w \in V \suchthat \set{v,w} \in E}$
the set of its \term{neighbors}
and $\deg(v) \df \card{N(v)}$ its \term{degree}.
A function \mbox{$\fn{c}{V \map \setn{k}}$} is called a $k$-\term{coloring}
or \term{coloring}, where $\setn{k} = \setft{1}{k}$.
Vertices of the graph represent players,
with the set of colors $\setn{k}$ being the strategy space for each player.
We sometimes call the set $\setn{k}$ the \term{spectrum}.
Clearly, $k$-colorings are exactly the strategy profiles of this game.\footnote{%
  The term \enquote{graph coloring game} in the literature also names a class of maker-breaker style games,
  which we are not referring to here.}
Given a coloring~$c$, define the \term{payoff} for player $v$ as
\begin{equation}
  \label{eq:def-payoff}
  \pay_v(c) \df \sum_{w \in N(v)} \apply{f}{\abs{c(v)-c(w)}} \comma
\end{equation}
where $f$ is a non-negative, real-valued function defined on $[0,k]$
(choosing this domain instead of $\setft{0}{k-1}$ is technically easier).
Given such $f$, we denote $f^* \df \max_{i \in \DIS} f(i)$
the maximum value that $f$ attains on the possible distances $\DIS \df \setft{0}{k-1}$ between two colors,
and
\begin{equation*}
  \DIS^*(f) \df \set{i \in \DIS \suchthat f(i) = f^*}
\end{equation*}
the set of distances where $f^*$ is attained.
We call $\applysmall{f}{\abs{c(v)-c(w)}}$ the \term{contribution}
of edge $\set{v,w}$.
So $f^*$ is an upper bound on the contribution of any edge,
and it is attained if the two players $v$ and $w$ manage to put a distance between each other
which is in the set $\DIS^*(f)$.
We assume $f^* > 0$, since otherwise the situation is uninteresting.
By concavity, this implies 
\begin{equation}
  \label{eq:f-positive}
  f(i) > 0 \quad\forall i \in \DIS^+ \df \setft{1}{k-1} \comma
\end{equation}
that is, $f$ is positive for all the positive distances.
\par
Let $c$ be a $k$-coloring.
For a player $v \in V$ and a color $t \in \setn{k}$
we write the coloring where $v$ 
\term{changes} to color $t$ as $\change{c}{v}{t}$, so
\begin{equation*}
  \change{c}{v}{t}(w) = \begin{cases}
    t & \text{if $w=v$} \\
    c(w) & \text{otherwise}
  \end{cases}
  \quad\quad \text{for $w \in V$.}
\end{equation*}
A coloring $c$ is called \term{stable} if
\begin{equation}
  \label{eq:def-stable}
  \pay_v(\change{c}{v}{t}) \leq \pay_v(c) \quad \forall t \in \setn{k}
  \quad \forall v \in V \period
\end{equation}
In other words, stable colorings are exactly
the pure Nash equilibria of the game.
The \term{social welfare} or \term{welfare}
of $c$ is $\welf(c) \df \sum_{v \in V} \pay_v(c)$.
We denote the optimum welfare as $\welf_{\OPT} \df \max_c \welf(c)$,
where $c$ runs over all $k$-colorings;
optimum colorings exist since there are only finitely many colorings.
The price of anarchy is:
\begin{equation*}
  \PoA(G,k,f) \df \max_{\text{$c$ is stable $k$-coloring}} \frac{\welf_{\OPT}}{\welf(c)}
\end{equation*}
This notion was introduced by Koutsoupias and Papadimitriou~\cite{KP99},
and the name \term{price of anarchy} was shortly after coined by Papadimitriou~\cite{Pap01}.
It measures the worst-case performance-loss due to non-cooperative behavior.
\par
It is easy to see that stable colorings always exist, by a potential function argument.
This has been observed before in the broader class of polymatrix common-payoff games~\cite{CD11},\footnote{%
  In~\cite{CD11}, the term \enquote{coordination} is used instead of \enquote{common-payoff}.}
and we repeat that simple argument here for completeness.
Note that we can write social welfare as follows:
\begin{equation}
  \label{eq:edgewise}
  \welf(c) = 2 \sum_{\set{u,w} \in E} \apply{f}{\abs{c(u)-c(w)}}
\end{equation}
If a player $v$ makes an \term{improvement step}
resulting in coloring $c'$
(\ie there is $t \in \setn{k}$ such that $c'= \change{c}{v}{t}$
and $\pay_v(c) < \pay_v(c')$) then
\begin{equation*}
  \sum_{\set{v,w} \in E} \apply{f}{\abs{c(v)-c(w)}}
  < \sum_{\set{v,w} \in E} \apply{f}{\abs{c'(v)-c'(w)}} \comma
\end{equation*}
that is, the part of the sum in \eqref{eq:edgewise} with edges incident in $v$ strictly increases.
All other terms in the sum are maintained.
Hence $\welf(c) < \welf(c')$.
So in particular any optimal coloring is stable.
Since existence of optimal colorings is guaranteed, we also have stable colorings.
Note however that there can be sub-optimal stable colorings,
which motivates the study of the price of anarchy.

\section{Previous, Related Work and Our Contribution}
\subsection*{Previous and Related Work}
Graph coloring has been a theme in Combinatorics
and Combinatorial Optimization for many decades.
A coloring $c$ for a graph $G=(V,E)$ is called \term{proper} or \term{legal} if
$c(v) \neq c(w)$ for all $\set{v,w} \in E$,
and given $v \in V$ we call a neighbor $w \in N(v)$ \term{properly colored} if $c(v) \neq c(w)$.
Determining the minimal $k$ such that a given graph admits a proper $k$-coloring,
\ie the \term{chromatic number} $\chi(G)$, is \NPhard~\cite{Kar72}.
One of the classical highlights is Brooks' theorem~\cite{Bro41,Lov75a},
stating that for each graph $G$ which is neither a complete graph nor a cycle of odd length, a proper $\Del$-coloring
can be constructed by a combinatorial algorithm in polynomial time,
where $\Del$ is the maximum vertex degree in $G$.
Another important algorithmic technique for finding proper colorings with provable worst-case performance
is semi-definite programming with randomized hyperplane rounding~\cite{KMS98}.
In addition, for more than three decades, distributed algorithms for proper colorings
have been studied, see \cite{BE13} for a recent survey.
In the distributed model, each vertex is a processor and can communicate with its neighbors,
where communication is done in rounds.
For example, in one round, each vertex could communicate its color to all of its neighbors.
In 2006, a study with human subjects was conducted by Kearns \etal~\cite{KSM06},
where each subject controlled the color of one vertex, to be selected from a set of $\chi(G)$ colors,
and each subject could see the colors of her neighbors.
The goal was to construct a proper coloring.
The study used graphs with $n=38$ vertices and focused on certain classes of graphs, like cycles, small world graphs,
and random graphs from the preferential attachment model.
It is reported how the time required to reach a proper coloring is influenced by the structure of the graph.
\par
A game-theoretic view on proper colorings was given in 2008 by Panagopoulou and Spirakis~\cite{PS08}.
In their model, payoff for a player $v$ is $0$ whenever there exists a non-properly colored neighbor of $v$,
otherwise payoff is the total number of players with color $c(v)$ in the graph.
The idea is to incentivize players to create a proper coloring with few different colors.
Indeed, two of the main results in \cite{PS08} are that Nash equilibria can be constructed
by improvement steps in polynomial time, and that they are proper colorings
using a total number of different colors that is upper-bounded by several known upper bounds on $\chi(G)$.
So \cite{PS08} yields a constructive proof for all of those bounds.
For the same model, improved bounds and an extension to coalitions were later given by Escoffier \etal~\cite{EGM10};
and Chatzigiannakis \etal \cite{CKPS10} gave algorithmic improvements including experimental studies.
Also in 2008, Chaudhuri \etal~\cite{CCJ08} considered the simpler payoff function
that is $1$ for player $v$ if all of $v$'s neighbors are properly colored and $0$ otherwise.
They showed that if the total number of available colors is $\Del+1$
and if players do improvement steps in a randomized manner, then an equilibrium (which also is a proper coloring)
is reached in $\Oh{\log(n)}$ steps with high probability.
\par
A generalization of proper colorings are \term{distance-constrained labelings},
which are connected to the frequency assignment problem, where colors correspond to frequencies~\cite{Jan02}.
Given integers $\eli{p}{l}$, a coloring $c$ is called an $L(\eli{p}{l})$-labeling if
for all $i \in \setn{l}$ and all $v,w \in V$ we have
$\abs{c(v)-c(w)} \geq p_i$ whenever $\dist_G(v,w) = i$,
where $\dist_G(v,w)$ is the graph-theoretical distance between $v$ and $w$,
\ie the length of a shortest $v$-$w$ path in $G$.
The notion of $L(1)$-labeling coincides with that of proper coloring.
The case of $L(p,q)$-labelings, that is, where $l=2$, has received special attention,
in particular in connection with frequency assignment, see, \eg \cite{Yeh06,Cal11,HIOU14}.
As a variation of this, not the distance $\abs{c(v)-c(w)}$ between colors is considered,
but instead it is required that
\begin{equation*}
  \min\set[big]{\abs{c(v)-c(w)}, \: k-\abs{c(v)-c(w)}} \geq p_i \comma
\end{equation*}
whenever \mbox{$\dist_G(v,w) = i$}, as considered in~\cite{HLS98}.
Another variation are $T$-colorings \cite{Hal80,Rob91}: given a set of integers $T$,
a $T$-coloring is one where $\abs{c(v)-c(w)} \not\in T$ whenever $\set{v,w} \in E$.
This problem also arises in frequency assignment
where $T$ is the set of forbidden separations between neighboring senders,
which are known to cause interference.
\par
In our model, we allow colorings being evaluated by the payoff function
and not only ask whether they have a certain property (\eg, being proper, a $T$-coloring, \etc) or not.
For example, using
\begin{equation}
  \label{eq:basic-payoff}
  \fnx{f}{[0,k] \map \RRnn}{x \mapsto
    \begin{cases}
      0 & \text{if $x=0$} \\
      1 & \text{otherwise}
    \end{cases}}
\end{equation}
as the function~$f$, payoff for each player $v$ is the number of her properly colored neighbors.
We call this \term{basic payoff}.
Note that this is a concave function.
With this payoff, the game is also known as a \term{max-$k$-cut game} (in the unweighted version),
since we partition the vertices in $k$ clusters and payoff for $v$ is
the total number of cut edges incident to $v$, \ie edges incident to $v$ and running between clusters.
These games were studied by Hoefer~\cite{Hoe07} in 2007,
and a tight bound of $\frac{k}{k-1}$ on the price of anarchy was proved,
where tightness is already attained on bipartite graphs.
The upper bound works by a mean-value argument:
for each coloring and for each player~$v$,
there is a color $t$ that occurs on at most $\frac{\deg(v)}{k}$ neighbors of $v$.
In a stable coloring, $v$ chooses this $t$, or even better if possible,
hence obtaining a payoff of at least $\frac{k-1}{k} \deg(v)$.
Since optimum welfare is upper-bounded by $2m$, the theorem follows from this
(using the handshaking lemma for the sum of degrees).
In \autoref{app:mean-value} we provide indication that already for $f(x)=x$,
a straightforward generalization of this technique cannot yield a constant upper bound.
\par
A weighted version of max-$k$-cut games is obtained by assigning a weight to each edge
and defining payoff for $v$ as the sum of the weights of the incident cut edges.
This version under the aspect of coalitions among players
was studied by Gourv{\`e}s and Monnot~\cite{GM09} in~2009.
\par
In 2013, Kun, Powers, and Reyzin~\cite{KPR13} considered complexity issues
for basic payoff, \ie, payoff as per \eqref{eq:basic-payoff}, combined with variations of the game.
It is clear that since welfare for basic payoff can only change
in integer steps and can never be more than $2m$,
we will reach a stable coloring after at most $\Oh{m}$ improvement steps
($m$ is the number of edges in the graph).
On the other hand, Kun \etal show that for basic payoff, it is \NPhard to decide whether a graph
admits a \term{strictly stable} coloring,
where the latter notion is defined by replacing $\leq$ for $<$ in \eqref{eq:def-stable}.
They also show that for basic payoff,
it is \NPhard to decide whether a \emphasis{directed} graph has a stable coloring,
where for directed graphs, payoff is defined by having the sum in \eqref{eq:def-payoff}
only range over the out-neighbors of~$v$.
\par
Recently, in 2014, Apt \etal~\cite{ARSS14} used the function
\begin{equation}
  \label{eq:coordination}
  \fnx{f}{[0,k] \map \RRnn}{x \mapsto
    \begin{cases}
      1 & \text{if $x=0$} \\
      0 & \text{otherwise}
    \end{cases}}
  \comma
\end{equation}
which counts the neighbors of the \emphasis{same} color,
together with the extension that each player has a set of colors to choose from
(whereas we in our model always allow all colors for all players).
They study this game under different aspects,
including coalitions, price of anarchy, and complexity.
As for the price of anarchy without coalitions,
it is easy to see that it can be unbounded.
The lower-bound construction depends on the fact that we can forbid certain colors for certain players:
for example, for each player $v$ let there be a distinct \enquote{private} color $s_v$,
and in addition there is a \enquote{common} color $t$.
Player $v$ can choose her color from the set $\set{s_v,t}$.
The social optimum, namely $2m$, is obtained when each player chooses $t$,
whereas a worst case stable coloring, with welfare $0$, is obtained when each player $v$ chooses $s_v$.
Without the ability to restrict players to certain colors, \ie if we use our framework for the function~$f$ in \eqref{eq:coordination},
a tight bound on the price of anarchy of $k$ is easy to see (\autoref{app:same-color}).
\par
The graph coloring games studied in our work
belong to the class of \term{polymatrix games} \cite{Yan68,CD11}.
In such a game, we have a graph $G=(V,E)$,
for each player a set of strategies, and for each edge $\set{v,w} \in E$ a two-player matrix game $\Gam_{\set{v,w}}$.
Each player $v$ chooses one strategy and has to play this same strategy in all the games $\set{\Gam_{\set{v,w}}}_{w \in N(v)}$
corresponding to incident edges.
Payoff for $v$ is the sum of the payoffs over all those two-player games.
A special case is that of \term{polymatrix common-payoff games},
which means that each $\Gam_{\set{v,w}}$ is a \term{common-payoff game},
\ie it always yields the same payoff for $v$ as for $w$.\footnote{%
  Such are called \enquote{coordination games} sometimes, but the use of this term is not consistent in the literature.
  We stick to the terminology from \cite[Sec.~1.3.2]{LS08}, which is \enquote{common-payoff game}.}
Thus our graph coloring games belong to this class,
since each edge $\set{v,w}$ contributes the same value $\applysmall{f}{\abs{c(v)-c(w)}}$ to the payoffs of $v$ and of $w$.
Very recently, in 2015, Rahn and Schäfer~\cite{RS15} studied polymatrix common-payoff games with coalitions.
They consider $(\al,l)$-equilibria, that are $\al$-approximate equilibria under coalitions of size $l$.
For the corresponding price of anarchy,
they give a lower bound of $2\al(n-1)/(l-1)+1-2\al$ and an upper bound of $2\al(n-1)/(l-1)$.
Note that in our work, we have $l=1$ since we do not consider coalitions,
and for this case their bounds are~$\infty$.
This follows already from the example given in~\cite{ARSS14},
which is based on restricting certain players to certain colors (the \enquote{private} and \enquote{public} colors)
and is explained in the previous paragraph.
\subsection*{Our Direction}
In 2013, Kun \etal \cite{KPR13}
named as open problems the study of the price of anarchy for payoff as per \eqref{eq:def-payoff}
and induced by $f$ being the identity, \ie the contribution of edge $\set{v,w}$
is the distance $\abs{c(v)-c(w)}$.
We call this \term{distance payoff}.
A variation, which Kun \etal also refer to,
is the notion studied by van den Heuvel \etal~\cite{HLS98} in the context of frequency assignment,
namely the contribution of $\set{v,w}$ is $\min\set{\abs{c(v)-c(w)}, \, k-\abs{c(v)-c(w)}}$,
which in our notation means $f(x) = \min\set{x, \, k-x}$.
We call this \term{cyclic payoff}.
It rewards players for keeping a \enquote{medium} distance from others.
This has an additional interpretation in connection with an example often given in the context of proper colorings,
namely where colors correspond to skills and people inside of an organization try to
develop skills that are complementary to nearby colleagues (see, \eg~\cite{KSM06,CCJ08}).
Cyclic payoff refines that idea: subjects are incentivized to develop different
\emphasis{but still related (that is, not too far away)} skills.
Distance and cyclic payoff both result from concave functions~$f$,
so they are both special cases of the payoff functions studied in our work.
\par
Spectrum sharing and frequency (or channel) assignment problems,
that is, when a multitude of participants compete for using the same or similar frequencies,
has received much attention lately, see, \eg~\cite{PZZ06,FK07,HHLM10,DA12}.
Many works in that field use some form of graph coloring (colors corresponding to frequencies)
in a distributed or game-theoretic setting.
It is common in the frequency assignment literature
to consider not only feasible versus infeasible colorings
but instead to quantify the \enquote{degree of interference}.
To this end,
the spectral distance $\abs{c(v)-c(w)}$ between players $v$ and $w$ is often used as a measure
or as a substantial ingredient to a measure.
In addition to the references we give above for distance-constrained labeling and for $T$-coloring,
this is documented for example in~\cite[Sec.~2]{Gam86}, \cite[Sec.~2.1]{ASH02},
\cite[Sec.~II.B]{SP97}, \cite[Sec.~3]{AHK+07}, and \cite[Sec.~11.6]{BEG+10}.
So our game can be used to model the case where each participant of a wireless network,
\eg a mobile telephone network or wireless sensor network,
chooses her frequency non-cooperatively and with respect to a measure of interference 
expressed by a concave function applied to spectral distances
(where higher values of that function mean less interference).
\par
A first game-theoretic study of distance and cyclic payoff in our framework
was conducted by Schink~\cite{Sch14} in 2014.
He observed that a stable $k$-coloring for distance payoff 
can be constructed from a stable $2$-coloring for basic payoff
by replacing color $2$ with~$k$.
A similar approach works for cyclic payoff if $k$ is even,
replacing color $2$ with $\frac{k}{2}+1$.
Since a stable coloring for basic payoff can be constructed
in $\Oh{m}$ improvement steps, 
this gives a runtime guarantee independent of~$k$.
Since welfare can reach up to $\Ohm{mk}$ for distance and cyclic payoff,
such a $k$-independent guarantee is not easily possible by
a direct argument using improvement steps.
Interestingly, for cyclic payoff and \emphasis{odd} $k$,
we have no $k$-independent runtime guarantee for the construction of stable colorings at this time.
For price of anarchy, Schink proved an upper bound of $\Del(G)$,
the maximum vertex degree in $G$, for cyclic payoff.
Apart from that, we are not aware of any bounds on the price of anarchy
for graph coloring games with concave payoff, not even conjectures.
The work by Rahn and Schäfer~\cite{RS15}, for general polymatrix common-payoff games,
can be considered orthogonal to ours since they do not consider the effects
of restricting the two-player games $\Gam_{\set{v,w}}$ to certain classes,
whereas we restrict to such games arising from applying a concave function
to the color distance $\abs{c(v)-c(w)}$,
resulting in small constant bounds on the price of anarchy.
Moreover, \cite{RS15} allows restricting players to certain sets of colors
(making finite bounds on the price of anarchy impossible without coalitions),
whereas in our model, each player has the same set of colors to choose from.
\subsection*{Our Contribution and Techniques}
We prove constant upper bounds on the price of anarchy for several classes of concave functions~$f$.
We prove a bound of $2$ for all concave functions $f$
which assume $f^*$ at a distance on or below $\floor{\frac{k}{2}}$,
that is, for which $\DIS^*(f) \cap \setft{0}{\floor{\frac{k}{2}}} \neq \emptyset$.
This includes cyclic payoff.
We show that for this class of functions, this bound is the best possible,
since for cyclic payoff with even $k$, the price of anarchy is exactly~$2$.
For non-decreasing concave functions, we also show an upper bound of $2$.
This includes distance payoff.
Again, we show that for this class of functions, this bound is the best possible.
Finally, for all remaining concave functions,
which have $\DIS^*(f) \cap \setft{\floor{\frac{k}{2}}+1}{k-2} \neq \emptyset$,
we prove an upper bound of~$3$.
The upper bounds are proved in \autoref{splitting} and \autoref{general},
and the lower bounds are given in \autoref{lower}.
\par
It may be surprising at first that the situation is not symmetric:
functions with their maximum left of the middle of the spectrum behave differently 
from functions with their maximum right of the middle.
But in fact this is to be expected since for example,
a player can always force all the distances to her neighbors
to be on or below $\floor{\frac{k}{2}}$ by choosing her own color as $\floor{\frac{k}{2}}+1$,
but it is not always possible to force all distances beyond~$1$.
So there is an asymmetry between short and long distances.
\par
All our proofs work by \emphasis{local} arguments.
That is, if we are to prove an upper bound of $\localparam \geq 1$
on the price of anarchy, we show the following:
for each player $v$,
given the colors $\set{c(w)}_{w \in N(v)}$ of her neighbors,
there is a color $t \in \setn{k}$ such that
\begin{equation*}
  \sum_{w \in N(v)} \apply{f}{\abs{c(w)-t}}
  \geq \frac{\deg(v) \cdot f^*}{\localparam} \period
\end{equation*}
Clearly, in a stable coloring,
each player chooses such a color $t$, or better.
Hence $\welf(c) \geq \frac{2 m f^*}{\localparam}$ for a stable $c$.
Since $\welf_{\OPT} \leq 2 m f^*$, the bound $\localparam$ follows.
\par
In the next section, we provide a framework for the type of local arguments described in the above paragraph.
To this end, we introduce the \term{local parameter} $\localparam(f,k)$ of a function $f$
and show $\PoA(G,f,k) \leq \localparam(f,k)$.
In \autoref{quick}, we show how to obtain an upper bound of $4$ on the local parameter, and thus on the price of anarchy,
for any concave function $f$ by a simple technique.
In \autoref{splitting}, a more elaborate technique is introduced.
It says that if we can find a representation of a number $\localparam \geq 1$
as a sum $\localparam_1 + \hdots + \localparam_k = \localparam$
with $\sum_{s=1}^k \localparam_s \cdot f(\abs{s-p}) \geq f^*$ for all $p \in \setn{k}$,
then we have $\localparam(f,k) \leq \localparam$,
and thus a bound of $\localparam$ on the price of anarchy.
Using this \term{splitting technique}, we prove our main results in \autoref{splitting} and \autoref{general}.
For example, the upper bound of $2$ for cyclic payoff can be proven using the splitting defined by
$\localparam_1 \df 1$ and $\localparam_{\floor{\frac{k}{2}}+1} \df 1$ and $\localparam_s \df 0$
for all remaining indices $s \in \setn{k} \setminus \set{1, \, \floor{\frac{k}{2}}+1}$.
\par
Appropriate splittings up to a certain granularity for moderate values of $k$ and given function $f$
can be found by computer, doing a simple enumeration.
We fix $\del=(\eli{\del}{r})$ such that $\localparam=\sum_{i=1}^r \del_i$,
for example $\del=(1,1,\frac{1}{4},\frac{1}{4})$ for $\localparam=2.5$.
For each $v=(\eli{v}{r}) \in \setn{k}^r$ define a splitting $\localparam(v)$ by
$\localparam_s(v) \df \sum_{\substack{i \in \setn{r} \\ v_i = s}} \del_i$ for each $s \in \setn{k}$.
Then we let the computer enumerate all $v \in \setn{k}^r$
and output those corresponding splittings $\localparam(v)$ that satisfy the necessary conditions.
Although this is for a concrete $k$ and not general, it can give valuable hints on how to do a general proof.
We used this, with $\del=(1,1)$ and $\del=(1,1,1)$, to obtain essential ideas for the proofs in \autoref{general}.
\subsection*{Ongoing Work and Open Problems}
\begin{compactitemize}
\item Computer experiments suggest that
  the splitting technique should be powerful enough to prove
  an upper bound of $2.5$ for general concave functions, instead of the $3$ that we get.
  A possible splitting appears to have the form $1+1+\frac{1}{4}+\frac{1}{4}$,
  but deciding on the actual indices, \ie, which $\eli{s}{4} \in \setn{k}$ to choose in order to define
  $\localparam_{s_1} \df 1$
  and $\localparam_{s_2} \df 1$
  and $\localparam_{s_3} \df \frac{1}{4}$
  and $\localparam_{s_4} \df \frac{1}{4}$, is more complicated.
  We verified by experiment that appropriate indices can be found at least for all $k \leq 100$.
\item Computer experiments suggest that
  for functions assuming their maximum left of the middle of the spectrum,
  better bounds than $2$ on the local parameter should possible,
  depending on the exact location of the maximum.
  However, the evidential basis for this is thin at this time,
  since enumerative computer experiments are hindered by the exponential search space
  (\cf the definition in \eqref{eq:anarchy-value} below).
\item Apart from the splitting technique, are there further theoretical methods to compute or to approximate the local parameter?
  Are there faster practical methods for this task than simple enumeration?
\item For which functions $f$ and which ranges for parameter $k$
  can we find graphs $G$ such that $\PoA(G,f,k) = \localparam(f,k)$?
  In this work in \autoref{lower},
  we construct such graphs for $f(x)=\min\set{x, \, k-x}$ and even $k$ (cyclic payoff)
  and for $f$ of the form $f(x) = ax+b$.
\item Computational issues for finding stable colorings should be addressed,
  in particular the construction of stable colorings
  for cyclic payoff and odd $k$ in a number of steps being independent of~$k$.
\end{compactitemize}

\section{The Local Parameter}
For a function $\fn{f}{[0,k] \map \RRnn}$ we define the \term{local parameter} of $f$ as
\begin{equation}
  \label{eq:anarchy-value}
  \localparam(f,k) \df \max_{\nn \in \NN} \, \max_{\eli{c}{\nn} \in \setn{k}} \,
  \frac{\nn f^*}{\max_{t \in \setn{k}} \sum_{i=1}^\nn \applysmall{f}{\abs{c_i-t}}}
  \period
\end{equation}
By \eqref{eq:f-positive}, the denominator is never zero.
Clearly, the denominator is never greater than $\nn f^*$, so $\localparam(f,k) \geq 1$.
The intuition is that we capture the best way a player can react to a given set of colors $\eli{c}{\nn}$
(which will be the colors of her neighbors when applying this to the graph coloring game)
relative to the maximum conceivable payoff~$\nn f^*$.
\begin{remark}
  We have $\PoA(G,k,f) \leq \localparam(f,k)$.
\end{remark}
\begin{proof}
  Let $c$ be stable and $v \in V$.
  By definition of the local parameter, we have
  \begin{equation*}
  \localparam(f,k) \geq
  \frac{\deg(v) \cdot f^*}{\max_{t \in \setn{k}} \sum_{w \in N(v)} \applysmall{f}{\abs{c(v)-t}}}
  \quad\text{hence,}\quad
  \max_{t \in \setn{k}} \sum_{w \in N(v)} \applysmall{f}{\abs{c(v)-t}} \geq
  \frac{\deg(v) \cdot f^*}{\localparam(f,k)} \period
  \end{equation*}
  That is, there is a color $t$ which $v$ can choose to obtain payoff at least that much.
  Since $c$ is stable, player $v$ chooses such a color, or better, so $\pay_v(c) \geq   \frac{\deg(v) \cdot f^*}{\localparam(f,k)}$.
  Choosing $c$ as a worst-case stable coloring,
  and using the trivial upper bound $\welf_{\OPT} \leq 2m f^*$ and the handshaking lemma, we obtain:
  \begin{equation*}
    \PoA(G,k,f) = \frac{\welf_{\OPT}}{\welf(c)} 
    \leq \frac{2mf^*}{\sum_{v \in V} \pay_v(c)}
    \leq \frac{2mf^*}{\frac{f^*}{\localparam(f,k)} \sum_{v \in V} \deg(v)}
    = \localparam(f,k)
    \tag*{\qedhere}
  \end{equation*}
\end{proof}
The local parameter has the following simple properties.
\begin{remark}\hfill
  \label{scaling-monotonicity}
  \begin{compactitemize}
  \item The local parameter is invariant against scaling, that is,
    $\localparam(f,k) = \localparam(\gam f,k)$ for all $\gam \in \RRpos$.
  \item The local parameter is monotone in the following way.
    Let $\fn{f,g}{[0,k] \map \RRnn}$ such that $f(i) \geq g(i)$ for all $i \in \DIS$, and $f^* = g^*$.
    Then $\localparam(f,k) \leq \localparam(g,k)$.
  \end{compactitemize}
\end{remark}
\begin{proof}
  Follows directly from the definition in \eqref{eq:anarchy-value}.
\end{proof}

\section{A First Upper Bound}
\label{quick}
As a start, we provide a lose bound of $4$ on the local parameter by a simple argument.
This bound is later improved to $2$ and $3$, depending on where $f$ assumes its maximum.
\begin{theorem}
  Let $f$ be concave. Then $\localparam(f,k) \leq 4$.
\end{theorem}
\begin{proof}
  Let $\nn \in \NN$ and $\eli{c}{\nn} \in \setn{k}$
  and define the function $\phi(t) \df \sum_{i=1}^\nn \applysmall{f}{\abs{c_i-t}}$.
  We are done if we can prove that there is $t \in \setn{k}$ with $\phi(t) \geq \frac{\nn f^*}{4}$.
  Let $k^* \in \DIS^*(f)$, \ie $k^* \in \setft{0}{k-1}$ with $f(k^*) = f^*$.
  By concavity, $f$ assumes a value of at least $\frac{f^*}{2}$
  between $k^*$ and the half-way point to $0$
  as well as to the half-way point to $k$, \ie
  for all $x \in H \df [\frac{k^*}{2}, \, k^*+\frac{k-k^*}{2}]
  = [\frac{k^*}{2}, \, \frac{k+k^*}{2}]$
  we have $f(x) \geq \frac{f^*}{2}$.
  One of the following two cases is given:
  $c_i \leq \floor{\frac{k}{2}}$ 
  for at least $\frac{\nn}{2}$ distinct $i \in \setn{\nn}$,
  or $c_i \geq \floor{\frac{k}{2}} + 1$
  for at least $\frac{\nn}{2}$ distinct $i \in \setn{\nn}$.
  \par
  Assume the first case.
  We choose $t \df \ceiling{\frac{k+k^*}{2}}
  \leq \ceiling{\frac{2k}{2}} = k$.
  Let $i \in \setn{\nn}$ be an index with $c_i \leq \floor{\frac{k}{2}}$
  (of which we have at least $\frac{\nn}{2}$ ones in this case).
  Then
  \begin{equation*}
    \abs{c_i-t} = t-c_i \leq t - 1 = \ceiling{\tfrac{k+k^*}{2}} - 1
    \leq \tfrac{k+k^*}{2} \comma
  \end{equation*}
  and
  \begin{equation*}
    t-c_i
    \geq \ceiling{\tfrac{k+k^*}{2}} - \floor{\tfrac{k}{2}}
    \geq \tfrac{k+k^*}{2} - \tfrac{k}{2} = \tfrac{k^*}{2} \period
  \end{equation*}
  Hence $\abs{c_i-t} \in H$,
  which means that $\applysmall{f}{\abs{c_i-t}} \geq \frac{f^*}{2}$.
  Since we have at least $\frac{\nn}{2}$ of those indices,
  it follows $\phi(t) \geq \frac{\nn}{2} \frac{f^*}{2} = \frac{\nn f^*}{4}$.
  \par
  The other case can be treated likewise;
  define $t \df \ceiling{\frac{k-k^*}{2}}$.
  Let $i \in \setn{\nn}$ be an index with $c_i \geq \floor{\tfrac{k}{2}} + 1$.
  Then
  \begin{equation*}
    \abs{c_i-t} = c_i-t
    \leq k - \ceiling{\tfrac{k-k^*}{2}}
    \leq k - \tfrac{k-k^*}{2}
    = \tfrac{k+k^*}{2} \comma
  \end{equation*}
  and
  \begin{equation*}
    c_i-t
    \geq \floor{\tfrac{k}{2}} + 1 - \ceiling{\tfrac{k-k^*}{2}}
    \geq \tfrac{k}{2} - \tfrac{1}{2} + 1 
    - \parens{\tfrac{k-k^*}{2} + \tfrac{1}{2}}
    = \tfrac{k^*}{2} \period
  \end{equation*}
  So also in the case, $\abs{c_i-t} \in H$.
\end{proof}

\section{The Splitting Technique}
\label{splitting}
We give a simple but powerful technique to prove upper bounds on the local parameter of a given function.
Given a number $\localparam \in \RR_{\geq 1}$,
we call a family of numbers $\eli{\localparam}{k} \in \RRnn$ a \term{splitting of $\localparam$},
provided that $\localparam = \sum_{s=1}^k \localparam_s$.
\begin{lemma}
  Let $\localparam \geq 1$ and $\eli{\localparam}{k} \in \RRnn$ be a splitting of $\localparam$.
  Let $\fn{f}{[0,k]\map\RRnn}$.
  Assume that the following condition is given:
  \begin{equation}
    \label{eq:lower-fstar}
    \tag{\textasteriskcentered}
    \forall p \in \setn{k} \bigholds
    \sum_{s=1}^k \localparam_s \cdot f(\abs{s-p}) \geq f^*
  \end{equation}
  Then $\localparam(f,k) \leq \localparam$.
\end{lemma}
\begin{proof}
  Let $\nn \in \NN$ and $\eli{c}{\nn} \in \setn{k}$.
  We prove that $\localparam \cdot \max_{t \in \setn{k}} \sum_{i=1}^\nn \applysmall{f}{\abs{c_i-t}} \geq \nn f^*$.
  For each $p \in \setn{k}$ denote $\nn_p \df \card{\set{i\in\setn{\nn}\suchthat c_i = p}}$,
  that is, how many times the number $p$ occurs in the family $\eli{c}{\nn}$.
  We have:
  \begin{align*}
    \localparam \cdot \max_{t \in \setn{k}} \sum_{i=1}^\nn \applysmall{f}{\abs{c_i-t}}
    & = \sum_{s=1}^k \localparam_s \cdot \max_{t \in \setn{k}} \sum_{i=1}^\nn \applysmall{f}{\abs{c_i-t}} & \text{def. of splitting} \\
    & \geq \sum_{s=1}^k \localparam_s \cdot \sum_{i=1}^\nn \applysmall{f}{\abs{c_i-s}} & \text{maximum} \\
    & = \sum_{s=1}^k \localparam_s \cdot \sum_{p=1}^k \nn_p \cdot \applysmall{f}{\abs{p-s}} & \text{subsume same values} \\
    & = \sum_{p=1}^k \nn_p \cdot \sum_{s=1}^k \localparam_s \cdot \applysmall{f}{\abs{p-s}} & \text{exchange summation} \\
    & \geq \sum_{p=1}^k \nn_p f^* & \text{by \eqref{eq:lower-fstar}} \\
    & = \nn f^*
    \tag*{\qedhere}
  \end{align*}
\end{proof}
We demonstrate the use of the splitting technique by a couple of simple proofs.
\begin{proposition}
  \label{affine}
  Let $a \in \RRpos$ and $b \in \RRnn$ and $f(x) = ax + b$.
  Then $\localparam(f,k) \leq \rho(a,b,k) \df 2 \frac{a (k-1) + b}{a (k-1) + 2b} \leq 2$.
\end{proposition}
\begin{proof}
  All we have to do is check \eqref{eq:lower-fstar} for this function $f$
  and an appropriate splitting of $\localparam = \rho(a,b,k)$.
  We have $f^* = a \, (k-1) + b$.
  Define $\localparam_1 \df \localparam_k \df \frac{\rho(a,b,k)}{2} = \frac{f^*}{a \, (k-1) + 2b}$
  and $\localparam_s \df 0$ for all other $s$, that is, all $s \in \setn{k} \setminus \set{1,k}$.
  We have to check
  \begin{equation*}
    \forall p \in \setn{k} \bigholds
    \frac{f^*}{a \, (k-1) + 2b} \cdot \parens[big]{f(p-1) + f(k-p)} \geq f^*
    \comma
  \end{equation*}
  that is,
  \begin{equation*}
    \forall p \in \setn{k} \bigholds
    f(p-1) + f(k-p) \geq a \, (k-1) + 2b
    \period
  \end{equation*}
  The latter is clearly true due to the definition of $f$.
\end{proof}
\begin{corollary}
  \label{cor:affine}
  Let $f$ be concave, non-constant, and non-decreasing.
  Then $\localparam(f,k) \leq \rho(a,b,k) \leq 2$, where $a \df \frac{f(k-1)-f(0)}{k-1}$
  and $b \df f(0)$.
  (Since $f$ is non-constant, $a>0$. The constant case is trivial and needs no attention.)
\end{corollary}
\begin{proof}
  Define the function $\fnx{g}{[0,k] \map \RRnn}{x \mapsto ax + b}$.
  By concavity of $f$, we have $g(x) \leq f(x)$ for all $x \in [0, \, k-1]$.
  By monotonicity of $f$, we have $f^* = g^*$.
  The corollary follows from \autoref{scaling-monotonicity}.
\end{proof}
We also treat the case of a decreasing affine and then a non-increasing concave function.
Here it makes sense to allow $f$ to assume negative values in $(k-1, \, k]$.
\begin{proposition}
  \label{decreasing}
  Let $a \in \RRpos$ and $b \in \RRnn$ and $f(x) = b - ax$, such that $f(k-1) \geq 0$.
  Then $\localparam(f,k) \leq \rho'(a,b,k) \df \frac{2b}{2b - a (k-1)}$,
  which is $2$ for the case of $f(k-1)=0$.
\end{proposition}
\begin{proof}
  We have $f^* = b$.
  Define $\localparam_1 \df \localparam_k \df \frac{\rho'(a,b,k)}{2} = \frac{f^*}{2b - a (k-1)}$
  and $\localparam_s \df 0$ for all other $s$.
  Now \eqref{eq:lower-fstar} follows from a simple calculation as in the proof of \autoref{affine}.
\end{proof}
\begin{corollary}
  \label{cor:decreasing}
  Let $f$ be concave, non-constant, and non-increasing.
  Then $\localparam(f,k) \leq \rho'(a,b,k)$, where $a \df \frac{f(0)-f(k-1)}{k-1}$
  and $b \df f(0)$.
\end{corollary}
\begin{proof}
  Like \autoref{cor:affine}.
\end{proof}
\begin{proposition}
  \label{cyclic}
  Let $f(x) = \min\set{x, \, k-x}$, that is, cyclic payoff.
  Then the price of anarchy is upper-bounded by $2$.
\end{proposition}
\begin{proof}
  All we have to do is check \eqref{eq:lower-fstar} for this function $f$
  and an appropriate splitting of $\localparam = 2$.
  For $i \in \NN$ denote $k_i \df \floor{\frac{k}{2}} + i$.
  Then $f^* = k_0$ and $f(x) = x$ if $x \leq k_0$ and $f(x) = k-x$ if $x \geq k_1$.
  Define $\localparam_1 \df 1$ and $\localparam_{k_1} \df 1$ and $\localparam_s \df 0$ for all other $s$.
  Condition~\eqref{eq:lower-fstar} reduces to:
  \begin{equation}
    \label{eq:cyclic-reduced}
    \forall p \in \setn{k} \bigholds
    f(p-1) + f(\abs{k_1-p}) \geq k_0
  \end{equation}
  To show \eqref{eq:cyclic-reduced}, let $p \in \setn{k}$.
  If $1 \leq p \leq k_1$, then $p-1 \leq k_0$ and $k_1-p \leq k_0$, so we have:
  \begin{equation*}
    f(p-1) + f(\abs{k_1-p})
    = f(p-1) + f(k_1-p) = p-1 + k_1-p = k_0
  \end{equation*}
  If $k_2 \leq p \leq k$, then $p-1 \geq k_2-1 = k_1$ and $p-k_1 \leq k - k_1 = \ceiling{\frac{k}{2}} - 1 \leq k_0$, so we have:
  \begin{equation*}
    f(p-1) + f(\abs{k_1-p})
    = f(p-1) + f(p-k_1)
    = k - (p-1) + p-k_1 = k - k_0 \geq k_0
  \end{equation*}
  This concludes the proof.
\end{proof}

\section{Lower Bounds}
\label{lower}
\begin{proposition}
  \label{lower-affine}
  The bound of $\rho(a,b,k)$ implied by \autoref{affine} on the price of anarchy for affine functions $f$ is the best possible,
  and the worst case is assumed already on bipartite graphs.
\end{proposition}
\begin{proof}
  We give an instance of the graph coloring game with function $f(x) = ax + b$
  that has price of anarchy $\rho(a,b,k)$.
  Consider the complete bipartite graph $K_{2,2}$
  and denote $\set{u_1,u_2}$ the vertices of one partition
  and $\set{w_1,w_2}$ those of the other (so edges are all $\set{u_i,w_j}$ with $i,j\in\set{1,2}$).
  Define coloring $c$ by:
  \begin{equation*}
    c(u_1) \df 1 \quad
    c(u_2) \df k \quad
    c(w_1) \df \floor{\tfrac{k+1}{2}} \quad
    c(w_2) \df \ceiling{\tfrac{k+1}{2}}
  \end{equation*}
  It is easy to see that $c$ is stable:
  players $w_1$ and $w_2$ have payoff $a \, (k-1) + 2 b$ each, no matter which color they choose.
  Players $u_1$ and $u_2$ also have payoff $a \, (k-1) + 2b$ each, but only for colors $1$ and $k$;
  for all other colors they get less.
  An optimal coloring is obtained by $c(u_i) \df 1$ and $c(w_i) \df k$ for $i \in \set{1,2}$,
  with each edge giving contribution $a\,(k-1)+b$.
  We have, using the number $m=4$ of edges,
  \begin{equation*}
    \frac{\welf_{\OPT}}{\welf(c)} 
    = \frac{2 \cdot 4 \cdot (a \, (k-1) + b)}{4 \cdot (a \, (k-1) + 2b)}
    = \rho(a,b,k)
    \period\tag*{\qedhere}
  \end{equation*}
\end{proof}
\begin{proposition}
  The bound of $\rho'(a,b,k)$ implied by \autoref{decreasing} on the price of anarchy
  for affine decreasing functions $f$ is the best possible,
  and the worst case is assumed already on bipartite graphs.
\end{proposition}
\begin{proof}
  We use the same graph as in the proof of \autoref{lower-affine}.
  However, we define coloring $c$ by:
  \begin{equation*}
    c(u_1) \df 1 \quad
    c(u_2) \df k \quad
    c(w_1) \df 1 \quad
    c(w_2) \df k
  \end{equation*}
  It is easy to see that this is stable with $\welf(c) = 4 \, (2b-a(k-1))$.
  Since the optimum is $8b$, when all players choose the same color,
  a price of anarchy of $\frac{2b}{2b-a(k-1)}=\rho'(a,b,k)$ follows.
\end{proof}
\begin{proposition}
  The bound of $2$ implied by \autoref{cyclic} on the price of anarchy for cyclic payoff is the best possible for even $k$,
  and the worst case is assumed already on a cycle of even length.
  For odd $k$, we have a lower bound of $\frac{3}{2} \parens{1-\frac{1}{k}}$.
\end{proposition}
\begin{proof}
  Again for each $i \in \NN_0$ denote $k_i \df \floor{\frac{k}{2}} + i$.
  First let $k$ be even.
  Consider a cycle of length $4n$ for some $n \in \NNone$
  and color like so:
  \begin{equation*}
    1,1,k_1,k_1,1,1,k_1,k_1,\hdots
  \end{equation*}
  Then half of the edges have contribution $0$, namely between players of the same color,
  and the other half has contribution $k_0$ each, so the welfare is $nk_0$.
  We prove that this coloring is stable.
  Let $v$ be a player with $c(v)=1$. Her payoff is $k_0$.
  If she changes to a color $2 \leq t \leq k_1$, her new payoff will be
  $(t-1) + (k_1-t) = k_1-1 = k_0$, so no improvement.
  If she changes to a color $k_1+1 \leq t \leq k$, her new payoff will be
  $(t-k_1) + k-(t-1) = k-k_1+1 = k_0 -1+1=k_0$, so also no improvement.
  The case $c(v)=k_1$ is treated likewise.
  An optimal coloring uses $1$ and $k_1$ alternately and yields welfare $2nk_0$.
  This proves the claim.\footnote{%
    The above construction is not stable for odd $k$,
    since then for example a player with color $k_1$ could change to $k_2$:
    this would not change the contribution of the edge to the $1$-colored neighbor (it remains $k_0$)
    but would increase distance from $0$ to $1$ regarding the $k_1$-colored neighbor, hence increasing payoff by~$1$.}
  \par
  For odd $k$, we take a cycle of length $6n$ for some $n \in \NNone$ and color like so:
  $1,k_1,k_2,1,k_1,k_2,\hdots$.
  The pattern $1,k_1,k_2$ can be repeated an integral number of times
  since the number of vertices is a multiple of~$3$.
  This yields welfare $2n \, (k_0+1+k_0)=4nk_0+2n$, so in comparison with the optimum
  (still attained by using $1$ and $k_1$ alternately, since number of vertices is even) we have
  \begin{equation*}
    \frac{6nk_0}{4nk_0+2n} = \frac{3}{2} \parens{1 - \frac{1}{2k_0+1}} = \frac{3}{2} \parens{1 - \frac{1}{k}} \period
  \end{equation*}
  We prove that this coloring is stable.
  Let $v$ be a player with $c(v)=1$. Her payoff is $2 k_0$.
  If she changes to color $t$ with $2 \leq t \leq k_1$, her new payoff will be
  $(k_1 - t) + (k_2 - t) = 2k_0 + 3 - 2t \leq 2k_0 - 1$, so no improvement.
  If she changes to color $t$ with $k_2 \leq t \leq k = 2k_0 + 1$, her new payoff will be
  $(t - k_1) + (t - k_2) = 2t - 2k_0 - 3 \leq 2 (2k_0+1) - 2k_0 - 3 = 2k_0-1$, so also no improvement.
  \par
  Now let $c(v)=k_1$. Her payoff is $k_0+1=k_1$.
  If she changes to color $t$ with $2 \leq t \leq k_1-1=k_0$, her new payoff will be
  $(t-1) + (k_2-t) = k_2 - 1 = k_1$, so it is no improvement.
  If she changes color to $1$, then her new payoff will be $k-(k_2-1)=k-k_1=k-k_0-1=k_1-1=k_0$,
  so no improvement; note that $k-k_0=k_1$.
  If she changes to color $t$ with $k_1 + 1 \leq t \leq k$, her new payoff will be
  $k-(t-1)+(t-k_2)=k+1-k_0-2=k_1-1$, so also no improvement.
  The case $c(v)=k_2$ can be treated likewise and is omitted here.
\end{proof}

\section{General Concave $f$}
\label{general}
We define a family of \enquote{prototype} concave functions.
For each $\ell \in \NN$ with $1 \leq \ell < k-1$ define:
\begin{equation}
  \label{eq:fl}
  \fnx{f_\ell}{[0, k] \map \RRnn}{x \mapsto
    \begin{cases}
      \tfrac{x}{\ell} & \text{if $x \leq \ell$} \\
      \tfrac{k-x}{k-\ell} & \text{if $x \geq \ell$}
    \end{cases}}
\end{equation}
So this function rises in a linear fashion from $0$ until it reaches value $1$ in $\ell$,
and then it drops in a linear fashion until it reaches value $0$ in $k$.
Clearly, $f^*_\ell = 1$.
For even $k$ and $\ell = \frac{k}{2}$, this function is that for cyclic payoff.
The following remark shows how to transfer bounds on the local parameter on members of this familiy
to general concave functions.
Note that the cases $0 \in \DIS^*(f)$ and $k-1 \in \DIS^*(f)$ describe monotone functions $f$
and have been covered in \autoref{splitting} already (monotone on $[0, \, k-1]$, which is sufficient).
\begin{remark}
  Let $\fn{f}{[0,k]\map\RRnn}$ be concave.
  Then $\localparam(f,k) \leq \localparam(f_\ell,k)$ for all $\ell \in \DIS^*(f) \setminus \set{0,k-1}$.
\end{remark}
\begin{proof}
  Let $\ell \in \DIS^*(f) \setminus \set{0,k-1}$.
  Define $g \df \frac{f}{f^*}$, then $g^* = 1$ and $g(\ell) = 1 = f_\ell(\ell)$,
  so $g^* = f_\ell^*$.
  By concavity,\footnote{%
    Due to the particular shape of $f_\ell$, in order for this concavity argument to work,
    it is important that $f$ is concave on $[0,k]$ and not only on $[0,k-1]$.}
  $f_\ell(x) \leq g(x)$ for all $x \in [0,k]$.
  Applying \autoref{scaling-monotonicity} two times yields:
  $\localparam(f,k) = \localparam\parens{\tfrac{f}{f^*}} = \localparam(g,k) \leq \localparam(f,k)$.
\end{proof}
\begin{theorem}
  \label{left}
  For $\ell < \ceiling{\frac{k}{2}}$ we have $\localparam(f_\ell,k) \leq 2$.
\end{theorem}
\begin{proof}
  For each $i \in \NN_0$ denote $k_i \df \floor{\frac{k}{2}} + i$
  and also $k' \df \ceiling{\frac{k}{2}}$.
  Define the following splitting: $\localparam_{k'} \df 1$ and $\localparam_{k'-\ell} \df 1$ and $\localparam_s \df 0$ for all other $s$.
  Let $p \in \setn{k}$. We have to show:
  \begin{equation*}
    \phi(p) \df f\parens{\abs{k'-p}} + f\parens{\abs{k'-\ell-p}} \geq 1
  \end{equation*}
  The following observations help to make the necessary case distinction:
  \begin{align*}
    \abs{k' - p} \leq \ell & \iff k' - \ell \leq p \leq k' + \ell \\
    \abs{k' - \ell - p} \leq \ell & \iff k' - 2 \ell \leq p \leq k'
  \end{align*}
  \begin{description}
  \item[Case $1 \leq p < k' - 2 \ell$:]
    \begin{equation*}
      \phi(p) 
      = \frac{k-(k'-p) + k-(k'-\ell-p)}{k-\ell}
      = \frac{2(k-k'+p)+\ell}{k-\ell}
      \geq \frac{2(k_0+1)+\ell}{k-\ell}
      \geq \frac{k+\ell}{k-\ell} > 1
    \end{equation*}
  \item[Case $k' - 2 \ell \leq p < k' - \ell$:]
    \begin{equation*}
      \phi(p) 
      = \frac{k-(k'-p)}{k-\ell} + \frac{k'-\ell-p}{\ell}
      \geq \frac{k-(k'-p)+k'-\ell-p}{k-\ell} = 1
    \end{equation*}
  \item[Case $k' - \ell \leq p \leq k'$:]
    \begin{equation*}
      \phi(p) = \frac{k'-p + p-k'+\ell}{\ell} = 1
    \end{equation*}
  \item[Case $k' < p \leq k'+\ell$:]
    \begin{equation*}
      \phi(p) = \frac{p-k'}{\ell} + \frac{k-(p-k'+\ell)}{k-\ell}
      \geq \frac{p-k'+k-(p-k'+\ell)}{k-\ell} = 1
    \end{equation*}
  \item[Case $k'+\ell < p \leq k$:]
    \begin{equation*}
      \phi(p) = \frac{k-(p-k') + k-(p-k'+\ell)}{k-\ell}
      = \frac{2 (k+k'-p) - \ell}{k-\ell}
      \geq \frac{2 k' - \ell}{k-\ell}
      \geq \frac{k - \ell}{k-\ell} = 1
      \tag*{\qedhere}
    \end{equation*}
  \end{description}
\end{proof}
In the following, we will treat the case of $\ell > \floor{\frac{k}{2}}$.
The only case thus not covered here is when $k$ is even and $\ell = \frac{k}{2}$,
but this has been covered in \autoref{splitting} since it coincides with cyclic payoff.
Hence so far we have $\localparam(f_\ell,k) \leq 2$ for all $\ell \leq \floor{\frac{k}{2}}$.
For the remaining cases, we start with a technical preparation.
\begin{proposition}
  \label{between}
  Let $\floor{\frac{k}{2}} < \ell \leq k-1$, in particular $k \geq 3$.
  Then there is an integer $i \in \NN$ with $2 \leq i \leq \ell$ such that:
  \begin{equation}
    \label{eq:between}
    \frac{(k-\ell) \, (2\ell-k)}{\ell}
    \leq i \leq
    \frac{\ell \, (k-\ell)}{2\ell-k} + 1
  \end{equation}
\end{proposition}
\begin{proof}
  For $\rho \in (\frac{1}{2}, 1]$ define:
  \begin{equation*}
    h_1(\rho) \df
    \frac{(1-\rho) \, k \, (2 \rho k - k)}{\rho k}
    = \frac{(1-\rho) \, (2 \rho - 1)}{\rho} \, k
    = \parens{3 - 2 \rho - \tfrac{1}{\rho}} \, k
  \end{equation*}
  and:
  \begin{equation*}
    h_2(\rho) \df
    \frac{\rho k \, (1-\rho) \, k}{2\rho k - k} + 1
    = \frac{\rho - \rho^2}{2\rho - 1} \, k + 1
    = \frac{1 - \rho}{2 - \frac{1}{\rho}} \, k + 1
  \end{equation*}
  That is, conceptually, we replace $\ell$ by $\rho k$.
  The value of $\rho$ corresponding to the maximum $\ell$, namely $k-1$, is $1-\frac{1}{k}$.
  \par
  It is easy to see that if there is an integer $i$ with \eqref{eq:between},
  then this $i$ can be chosen so that $2 \leq i \leq \ell$:
  we have $h_1(\rho) \leq \rho k$ for all $\rho$ (since the function $\rho \mapsto 3 - 3 \rho - \frac{1}{\rho}$
  is concave with a negative value at its maximum $\rho = \frac{1}{\sqrt{3}}$);
  and moreover $h_2(\rho) \geq 2$ for all $\rho \leq 1-\frac{1}{k}$ since $h_2(1-\frac{1}{k}) \geq 2$ and $h_2$ is non-increasing.
  \par
  Finall we note that $h(\rho) \df h_2(\rho) - h_1(\rho) - 1 \geq 0$ for all $\rho$;
  in other words, we show that there is space for at least one integer between the two bounds.
  This follows from the three facts: $h(1) = 0$ and $h$ is non-increasing (seen by looking at its first derivative)
  and $h(\rho) \tends \infty$ as $\rho \tends \frac{1}{2}$.
\end{proof}
\begin{theorem}
  For $\ell > \floor{\frac{k}{2}}$ we have $\localparam(f_\ell,k) \leq 3$.
\end{theorem}
\begin{proof}
  We choose $2 \leq i \leq \ell$ as per \autoref{between}.
  Then define $\localparam_1 \df 1$ and $\localparam_i \df 1$ and $\localparam_{\ell+1} \df 1$.
  Let $p \in \setn{k}$. We have to show:
  \begin{equation*}
    \phi(p) \df \applysmall{f}{p-1} + \applysmall{f}{\abs{p-i}} + \applysmall{f}{\abs{p-(\ell+1)}} \geq 1
  \end{equation*}
  \begin{description}
  \item[Case $1 \leq p \leq i$:]
    \begin{equation*}
      \phi(p) = \frac{(p-1)+(i-p)+(\ell+1-p)}{\ell} = \frac{i+\ell-p}{\ell} \geq \frac{i+\ell-i}{\ell} = 1
    \end{equation*}
  \item[Case $i+1 \leq p \leq \ell+1$:]
    \begin{equation*}
      \phi(p) = \frac{(p-1)+(p-i)+(\ell+1-p)}{\ell} = \frac{p-i+\ell}{\ell} \geq \frac{(i+1)-i+\ell}{\ell} > 1
    \end{equation*}
  \item[Case $\ell+2 \leq p \leq \ell+i$:]
    We have
    \begin{equation*}
      \phi(p) 
      = \frac{k-(p-1)}{k-\ell} + \frac{p-i}{\ell} + \frac{p-(\ell+1)}{\ell}
      = \frac{k-p+1}{k-\ell} + \frac{2p-i-\ell-1}{\ell}
      \period
    \end{equation*}
    It follows:
    \begin{align*}
      \phi(p) \geq 1
      & \iff \ell \, (k-p+1) + (k-\ell) \, (2p-i-\ell-1) - (k-\ell) \, \ell \geq 0 \\
      & \iff p \, (2k-3\ell) + \ell - (k-\ell) \, (i+\ell+1) + \ell^2 \geq 0 \\
      & \iff p \, (2k-3\ell) - (k-\ell) \, (i+\ell+1) + \ell \, (\ell+1) \geq 0 \\
      & \iff p \, (2k-3\ell) + (\ell+1) \, (2\ell-k) \geq (k-\ell) \, i
    \end{align*}
    If $2k \geq 3\ell$, then we plug in the smallest value for $p$:
    \begin{align*}
      \phi(p) \geq 1
      & \impliedby (\ell+2) \, (2k-3\ell) + (\ell+1) \, (2\ell-k) \geq (k-\ell) \, i \\
      & \impliedby (\ell+1) \, (2k-3\ell) + (\ell+1) \, (2\ell-k) \geq (k-\ell) \, i \\
      & \iff (\ell+1) \, (2k-3\ell+2\ell-k) \geq (k-\ell) \, i \\
      & \iff \ell+1 \geq i
    \end{align*}
    The last condition is true by the choice of $i$, hence $\phi(p) \geq 1$.
    \par
    On the other hand, if $2k < 3\ell$, then we plug in the greatest value for $p$:
    \begin{align*}
      \phi(p) \geq 1
      & \impliedby (\ell+i) \, (2k-3\ell) + (\ell+1) \, (2\ell-k) \geq (k-\ell) \, i \\
      & \iff \ell \, (2k-3\ell) + (\ell+1) \, (2\ell-k) \geq (2\ell-k) \, i \\
      & \iff \frac{\ell \, (k-\ell)}{2\ell-k} + 1 \geq i
    \end{align*}
    The last condition is true by the choice of $i$, hence $\phi(p) \geq 1$.
  \item[Case $\ell+i+1 \leq p \leq k$:]
    We have
    \begin{equation*}
      \phi(p)
      = \frac{k-(p-1)}{k-\ell} + \frac{k-(p-i)}{k-\ell} + \frac{p-(\ell+1)}{\ell}
      = \frac{2(k-p)+i+1}{k-\ell} + \frac{p-\ell-1}{\ell}
      \period
    \end{equation*}
    Using $k < 3 \ell$, it follows:
    \begin{align*}
      \phi(p) \geq 1
      & \iff 2 \ell \, (k-p) + \ell \, (i + 1) + (k-\ell) \, (p-\ell-1) - (k-\ell) \, \ell \geq 0 \\
      & \iff p \, (k-3\ell) + 2 \ell k + \ell \, (i + 1) - (k-\ell) \, (2\ell+1) \geq 0 \\
      & \impliedby k \, (k-3\ell) + 2 \ell k + \ell \, (i + 1) - (k-\ell) \, (2\ell+1) \geq 0 \\
      & \iff k \, (k-\ell) + \ell \, (i + 1) - (k-\ell) \, (2\ell+1) \geq 0 \\
      & \iff \ell \, (i + 1) \geq (k-\ell) \, (2\ell-k+1) \\
      & \iff i \geq \frac{(k-\ell) \, (2\ell-k+1)}{\ell} - 1 \\
      & \impliedby i \geq \frac{(k-\ell) \, (2\ell-k)}{\ell}
    \end{align*}
    The last condition is true by the choice of $i$, hence $\phi(p) \geq 1$.\qedhere
  \end{description}
\end{proof}

\appendix

\newpage\noindent\Large\textbf{Appendix}\normalsize

\section{Upper Bound by Mean-Value Argument}
\label{app:mean-value}

We prove a rough bound on the price of anarchy for distance payoff (that is, $f(x)=x$)
using a straightforward generalization of a mean-value argument from the proof of \cite[Prop.~2]{KPR13}.
We believe that not much better bounds than this can be obtained without extending the technique.

\begin{proposition}
  \label{mean-value-2k}
  The price of anarchy for distance payoff is upper-bounded by~$2k$.
\end{proposition}
\begin{proof}
  Let $c$ be a stable $k$-coloring for $G$ and fix a player $v$.
  Denote
  \begin{equation}
    \label{eq:ell}
    \ell \df \max_{t\in\setn{k}} \, \card{\set{w \in N(v) \suchthat c(w) = t}}
  \end{equation}
  the cardinality of the largest color class in $v$'s neighborhood.
  Then clearly $\ell \geq \ceiling{\frac{\deg(v)}{k}}$.
  Let $t$ be a color where the maximum in \eqref{eq:ell} is attained,
  so there are $\ell$ neighbors of $v$ with color $t$.
  At least one of the two numbers, $t+\floor{\frac{k}{2}}$ or $t-\floor{\frac{k}{2}}$ is in $\setn{k}$.
  By choosing an appropriate one of them, $v$ puts distance $\floor{\frac{k}{2}}$
  between herself and those $\ell$ neighbors, so each of them will contribute $\floor{\frac{k}{2}}$ to $v$'s payoff.
  In a stable coloring, such as $c$, player $v$ chooses such color or better,
  hence
  \begin{equation*}
    \welf(c) \geq \sum_{v \in V} \tfrac{\deg(v)}{k} \floor{\tfrac{k}{2}} = \tfrac{2m}{k} \floor{\tfrac{k}{2}}
    \geq \tfrac{2m}{k} \tfrac{k-1}{2} \period
  \end{equation*}
  Here, $m$ is the number of edges in $G$.
  Using the trivial upper bound $\welf_{\OPT} \leq 2m \, (k-1)$, we obtain:
  \begin{equation*}
    \frac{\welf_{\OPT}}{\welf(c)} \leq \frac{2m \, (k-1) \cdot 2k}{2m \, (k-1)} = 2k \tag*{\qedhere}
  \end{equation*}
\end{proof}

\section{Counting Neighbors with Same Color}
\label{app:same-color}

\begin{proposition}
  Define $f$ as counting the neighbors with same color, that is,
  \begin{equation*}
    \fnx{f}{[0,k] \map \RRnn}{x \mapsto
      \begin{cases}
        1 & \text{if $x=0$} \\
        0 & \text{otherwise}
      \end{cases}}
    \period
  \end{equation*}
  Then the price of anarchy with respect to $f$ is upper-bounded by $k$, and this bound is tight.
\end{proposition}
\begin{proof}
  The upper bound follows from a mean-value argument like in the proof of \autoref{mean-value-2k}:
  for each player $v$, there is one color with which at least $\ceil{\frac{\deg(v)}{k}}$
  of her neighbors are colored, so choosing this color will yield at least that much payoff for $v$.
  Hence $\welf(c) \geq \sum_{v \in V} \frac{\deg(v)}{k} = \frac{2m}{k}$ for each stable coloring $c$.
  Using the trivial upper bound $\welf_{\OPT} \leq 2m$ yields the claim.
  \par
  For the lower bound, consider the complete bipartite graph $K_{k,k}$.
  Clearly, $\welf_{\OPT} = 2 k^2$, which is attained if all players choose the same color, for example color~$1$.
  Enumerate vertices in one partition $\set{\eli{v}{k}}$ and in the other $\set{\eli{w}{k}}$
  and define $c(v_i) \df c(w_i) \df i$ for each $i \in \setn{k}$.
  Then $c$ is stable since whatever color a player chooses, her payoff is always $1$.
  We have $\welf(c) = 2 k$, and the price of anarchy is at least $\frac{2k^2}{2k}=k$.
\end{proof}

\end{document}